\documentclass[runningheads]{llncs}

\usepackage[T1]{fontenc}

\usepackage{amsmath,amssymb,amsfonts}
\usepackage{graphicx}

\usepackage{xcolor}
\usepackage{fontawesome}
\usepackage{booktabs}
\usepackage[normalem]{ulem}

\usepackage{comment}
\usepackage{url}

\usepackage{tabularx}

\usepackage{bbding}

\newcommand{\ie}{\textit{i}.\textit{e}., }
\newcommand{\eg}{\textit{e}.\textit{g}., }

\def\BibTeX{{\rm B\kern-.05em{\sc i\kern-.025em b}\kern-.08em
    T\kern-.1667em\lower.7ex\hbox{E}\kern-.125emX}}

\usepackage{orcidlink}

\begin{document}

\title{Hard and Soft Real-Time DNN Inference Scheduling on Heterogeneous Multi-Accelerator Systems via Reinforcement Learning}
\titlerunning{Hard and Soft Real-Time DNN Inference Scheduling}

\author{
Francesco~Giulio~Blanco\,\orcidlink{0009-0003-6700-8111}%
\thanks{Corresponding author.\\This paper has been accepted for presentation at the 34th International Conference on Real-Time Networks and Systems (RTNS 2026).} \and\\
Enrico~Russo\,\orcidlink{0000-0002-7598-146X} \and
Giuseppe~Ascia\,\orcidlink{0000-0001-7452-5828} \and
Davide~Patti\,\orcidlink{0000-0003-0874-7793} \and
Maurizio~Palesi\,\orcidlink{0000-0003-3129-0664}
}

\authorrunning{F. G. Blanco et al.}

\institute{
University of Catania, Catania, Italy\\
\email{francesco.blanco@phd.unict.it},
\email{\{name.surname\}@unict.it}
}

\maketitle

\begin{abstract}
We present a scheduling framework for heterogeneous multi-accelerator systems executing both hard real-time and soft real-time QoS-aware DNN inference workloads. The framework combines an online reinforcement learning (RL) scheduler with a deterministic runtime safety layer. The RL scheduler dynamically generates candidate schedules by jointly selecting sub-job priorities and accelerator assignments, while runtime certification guarantees that all HRT deadlines remain satisfied independently of the learned policy.

To enable runtime certification in heterogeneous precedence-constrained systems, we introduce a schedulability-preserving prefix replacement criterion based on machine-wise cuts. The proposed criterion guarantees that runtime-generated schedule prefixes preserve the existence of a feasible continuation without requiring complete finish-time dominance with respect to a feasible reference schedule.

Experimental results on heterogeneous DNN accelerator platforms show that the proposed framework preserves HRT schedulability while improving QoS-aware deadline satisfaction, client-level service quality, and accelerator utilization compared to state-of-the-art scheduling approaches.

\keywords{Real-time \and Scheduling \and DNN \and Reinforcement learning \and QoS \and Heterogeneous \and Multi-tenant \and Multi-accelerator \and Multi-processor}
\end{abstract}

\section{Introduction}

Deep Neural Network (DNN) inference is increasingly used in safety-critical cyber-physical systems such as autonomous vehicles, robotics, industrial automation, and edge AI infrastructures \cite{aromolo2025real,restuccia2021time}. These systems often execute mixed workloads on heterogeneous multi-accelerator platforms, where hard real-time (HRT) inference tasks coexist with soft real-time (SRT) requests generated by multiple clients \cite{vestal2007preemptive,yu2021automated,ling2021rt}.

Scheduling such workloads is challenging because DNN inference induces precedence-constrained execution graphs, accelerator heterogeneity couples allocation and priority decisions, and SRT requests arrive dynamically at runtime \cite{romero_usenix21,li2021throughput}. At the same time, HRT tasks must always satisfy their deadlines under worst-case execution assumptions.

Existing real-time scheduling approaches can provide timing guarantees, but are often based on static schedules that poorly adapt to runtime variability \cite{kwon2021heterogeneous,das2024multi}. Conversely, adaptive online schedulers and multi-tenant inference systems improve utilization and responsiveness, but generally lack strict HRT guarantees \cite{choi2020prema,RELMAS1}. Reinforcement Learning (RL) is attractive for online scheduling because it can learn workload-dependent policies without handcrafted heuristics \cite{RELMAS2}, and this is the reason why it is employed in this work, but learned policies alone cannot be trusted for safety-critical scheduling.

This paper presents a scheduling framework for heterogeneous multi-accelerator systems that combines adaptive online scheduling with strict HRT guarantees. The framework separates optimization from safety. An RL-based online scheduler dynamically generates candidate schedule prefixes\footnote{Before providing the formal definition of a schedule prefix in Definition~\ref{def:scheduleprefix}, we informally clarify the terminology. Given a vector of local times, one for each machine in the multi-accelerator system, the corresponding schedule prefix is the portion of the schedule preceding such time on each machine.} by jointly selecting accelerator assignments and execution priorities, while a deterministic runtime safety layer validates each candidate before execution.

To enable runtime certification in heterogeneous precedence-constrained systems, we introduce a schedulability-preserving prefix replacement criterion based on machine-wise cuts, inspired by transfer schedulability \cite{willemsen2025transfer}. The proposed criterion guarantees that runtime-generated schedule prefixes preserve the existence of a feasible continuation without requiring complete finish-time dominance with respect to a feasible reference schedule.

Experimental results on heterogeneous DNN accelerator platforms show that the proposed framework preserves HRT schedulability while improving SRT deadline satisfaction, client-level QoS, and accelerator utilization compared to state-of-the-art scheduling approaches.

The main contributions of this paper are:
\begin{itemize}
\item A formulation of mixed hard/soft real-time DNN inference scheduling on heterogeneous multi-accelerator systems using a DAG-based task model.
\item A schedulability-preserving runtime certification mechanism for heterogeneous precedence-constrained schedules.
\item An RL-based online scheduler for joint accelerator allocation and execution priority assignment within a deterministic HRT safety envelope. 
\\Our open-source implementation is available at: \\\url{https://github.com/BlancoFrancescoGiulio/retnas}. 
\item Formal proofs showing that runtime certification preserves HRT schedulability independently of the employed policy (\textit{e.g.}, RL).
\item An experimental evaluation demonstrating improved SRT QoS and accelerator utilization while preserving HRT correctness.
\end{itemize}

\section{System and Scheduling Model}
\label{sec:model}

\begin{table}[t]
\centering
\caption{Summary of notation}
\label{tab:notation}
\renewcommand{\arraystretch}{1.05}
\footnotesize

\begin{minipage}[t]{0.485\columnwidth}
\vspace{0pt}
\begin{tabularx}{\linewidth}{@{}lX@{}}
\toprule
\textbf{Symbol} & \textbf{Description} \\
\midrule
$\mathcal{M}$, $M$ & Set and number of processors/sub-accelerators \\
$\mathcal{M}_v$ & Processors eligible for sub-job $v$ \\
$C_{v,m}$ & WCET of sub-job $v$ on processor $m$ \\
$s_v$, $f_v$ & Start and finish time of sub-job $v$ \\
$f_J$ & Completion time of DAG job $J$ \\
\midrule
$G_J=(V_J,E_J)$ & DAG representation of job $J$ \\
$V_J$, $E_J$ & Sub-jobs and precedence constraints of $J$ \\
\midrule
$\mathcal{T}^{\mathrm H}$ & Set of HRT tasks \\
$\tau_i=(G_i,D_i,T_i)$ & HRT task $i$ \\
$J_{i,k}^{\mathrm H}$ & $k$-th HRT DAG job of $\tau_i$ \\
$a_{i,k}$, $d_{i,k}^{\mathrm H}$ & Release time and deadline of $J_{i,k}^{\mathrm H}$ \\
\bottomrule
\end{tabularx}
\end{minipage}
\hfill
\begin{minipage}[t]{0.485\columnwidth}
\vspace{0pt}
\begin{tabularx}{\linewidth}{@{}lX@{}}
\toprule
\textbf{Symbol} & \textbf{Description} \\
\midrule
$J_j^{\mathrm S}$ & SRT DAG job $j$ \\
$a_j$, $d_j^{\mathrm S}$ & Arrival time and soft deadline of $J_j^{\mathrm S}$ \\
$\kappa_j$ & Client generating $J_j^{\mathrm S}$ \\
$\mathcal K$ & Set of clients \\
$H^{\mathrm S}$ & System-level SRT deadline-hit ratio \\
$H_k^{\mathrm S}$ & Client-level SLO achievement rate \\
$\hat H_k^{\mathrm S}$ & Target SLO achievement rate of client $k$ \\
$R_{J_j^{\mathrm S}}$ & Response time of SRT job $J_j^{\mathrm S}$ \\
\midrule
$S^{\mathrm{ref}}$ & Offline feasible WCET reference schedule \\
$S^{\mathrm{cand}}$ & Online candidate schedule\\
$T_S$ & Runtime scheduler activation period \\
$\tau_{\mathrm{sch}}$ & Upper bound on runtime scheduling overhead \\
\bottomrule
\end{tabularx}
\end{minipage}

\end{table}

We consider the problem of scheduling periodic hard real-time (HRT) tasks
together with aperiodic soft real-time (SRT) jobs on a heterogeneous processing
platform.
The model follows standard
real-time scheduling terminology and DAG-based task models
\cite{buttazzo1997hard,baruah2015multiprocessor,davis2011survey,tang2011combining,dag}.

\subsection{Platform and DAG Job Model}
\label{subsec:platform-dag-model}

The platform consists of a finite set of machines (processors or sub-accelerators),
$\mathcal{M}=\{1,\ldots,M\}$. The platform is heterogeneous: the execution time
of a computation depends on both the computation and the processor on which it
executes. Thus, processors are modeled as unrelated machines.

Computations are represented as DAG jobs. A DAG job $J$ is described by
$G_J=(V_J,E_J)$, where each vertex $u\in V_J$ is a sequential computation,
called a sub-job, and each edge $(u,v)\in E_J$ is a precedence constraint.
A sub-job $u$ is called a \emph{terminal} sub-job
if it has no successors, i.e.,
$\nexists v \in V_J : (u,v)\in E_J$.
If
sub-job $u$ starts at time $s_u$ and finishes at time $f_u$, then
$(u,v)\in E_J$ implies $s_v\ge f_u$. The completion time of job $J$ is
$f_J=\max_{u\in V_J} f_u$.
For each sub-job $u$, let $\mathcal{M}_u\subseteq\mathcal{M}$ denote the set of
processors on which $u$ is eligible to execute. For each eligible processor
$m\in\mathcal{M}_u$, $C_{u,m}$ denotes the worst-case execution time (WCET) of
$u$ on $m$; for $m\notin\mathcal{M}_u$, we set $C_{u,m}=\infty$. If $u$ is
assigned to processor $m_u$, then under worst-case execution $f_u=s_u+C_{u,m_u}$.

\paragraph{Communication and data-movement assumptions.}
Prior real-time DNN scheduling approaches explicitly model fine-grained
data-movement effects, including cross-processor communication, CPU/GPU
data transfers, and layer-level pipelining or migration
\cite{ling2022blastnet,kang2021lalarand,xiang2019pipelined}.
In this work, we incorporate data-movement
costs into the profiled execution-time bounds $C_{u,m}$. In particular,
for HRT sub-jobs, $C_{u,m}$ conservatively accounts for the activation
write-back and read costs required for the execution of sub-job $u$ on
accelerator $m$.

The runtime scheduler, instead, can exploit locality when two layers with a data dependency or two instances of the same layer execute sequentially
on the same accelerator, by reusing locally available activations or
weights. However, the HRT reference schedule construction and the runtime certification described in Sections~\ref{subsec:reference_schedule} and \ref{subsec:acceptance_problem} do not rely on such reuse. Indeed, after a runtime
prefix replacement, another HRT or SRT sub-job may execute between two
consecutive layers, or the successor may be assigned to a different
accelerator, invalidating assumptions on local data availability.

\subsection{Hard and Soft Real-Time Workloads}
\label{subsec:workloads}

The HRT workload is a finite set of periodic DAG tasks
$\mathcal{T}^{\mathrm H}=\{\tau_1,\ldots,\tau_n\}$. Each task is described by
\mbox{$\tau_i=(G_i,D_i,T_i)$}, where $G_i=(V_i,E_i)$ is its DAG, $D_i$ is its relative
deadline, and $T_i$ is its period. We consider constrained-deadline tasks, i.e.,
$D_i\le T_i$.

The $k$-th release of task $\tau_i$ generates an HRT DAG job~$J_{i,k}^{\mathrm H}$ at time $a_{i,k}=a_{i,0}+kT_i$, with absolute deadline
\mbox{$d_{i,k}^{\mathrm H}=a_{i,k}+D_i$}. An HRT job meets its deadline iff
\mbox{$f_{J_{i,k}^{\mathrm H}}\le d_{i,k}^{\mathrm H}$}. The HRT requirement is that
every generated job of every admitted HRT task meets its deadline.

The SRT workload consists of aperiodic DAG jobs arriving online. Each SRT job is
modeled as \mbox{$J_j^{\mathrm S}=(G_j,a_j,d_j^{\mathrm S},\kappa_j)$}, where
\mbox{$G_j=(V_j,E_j)$} is its DAG, $a_j$ is its arrival time, $d_j^{\mathrm S}$ is its
soft absolute deadline, and $\kappa_j\in\mathcal K$ identifies the client that
generated the request. Unlike HRT tasks, SRT jobs are not known during admission
control and no minimum inter-arrival time is assumed.

For SRT jobs, deadline misses do not compromise correctness, but degrade
QoS. If $f_{J_j^{\mathrm S}}>d_j^{\mathrm S}$, the request
misses its soft deadline. SRT jobs may therefore be delayed, canceled, or
rejected whenever necessary to preserve the HRT guarantee. For a set
$\mathcal{J}^{\mathrm S}$ of requested SRT jobs, the system-level soft-deadline
hit ratio is
\[
H^{\mathrm S}
=
\frac{
|\{J_j^{\mathrm S}\in\mathcal{J}^{\mathrm S}
\mid f_{J_j^{\mathrm S}}\le d_j^{\mathrm S}\}|
}{
|\mathcal{J}^{\mathrm S}|
}.
\]
The response time of an SRT job is
$R_{J_j^{\mathrm S}}=f_{J_j^{\mathrm S}}-a_j$.

\subsection{Client-Level QoS and Fairness}
\label{subsec:client-qos}

The system-level hit ratio $H^{\mathrm S}$ is an average system-level QoS metric and may hide
severe client-level QoS degradation; as observed in real-time systems, averages offer no guarantees~\cite{art:buttazzo_1999qos}. For example,
one client may experience a high deadline-hit ratio while another receives very
poor service, even though $H^{\mathrm S}$ remains acceptable.

We therefore define client-level Service-Level Objective (SLO) achievement rates, which represent client-level QoS metrics. Let
$\mathcal{J}_k^{\mathrm S}=\{J_j^{\mathrm S}\in\mathcal{J}^{\mathrm S}\mid
\kappa_j=k\}$ be the set of SRT jobs requested by client
$k\in\mathcal K$. The SLO achievement rate of client $k$ is
\[
H_k^{\mathrm S}
=
\frac{
|\{J_j^{\mathrm S}\in\mathcal{J}_k^{\mathrm S}
\mid f_{J_j^{\mathrm S}}\le d_j^{\mathrm S}\}|
}{
|\mathcal{J}_k^{\mathrm S}|
}.
\]
Each client may specify a target SLO achievement rate $\hat H_k^{\mathrm S}$.
The scheduler should therefore improve not only the system-level SRT hit ratio, but
also the client-level satisfaction with respect to these targets. A compact client-oriented QoS objective is to maximize
\[
\min_{k\in\mathcal K}
\left(H_k^{\mathrm S}-\hat H_k^{\mathrm S}\right),
\]
subject to preserving all HRT deadlines. In cloud-oriented deployments, a client
may correspond to a service tenant; in edge and cyber-physical deployments, a
client may correspond to, \eg a robot, device, sensor stream, or independent
inference pipeline.

\subsection{Scheduling Model and Objective}
\label{subsec:scheduling-model}

As in standard multiprocessor scheduling, scheduling consists of two coupled
problems~\cite{davis2011survey}: an allocation problem, which decides on which
processor a computation executes, and a priority problem, which decides when and
in what order ready computations execute. In system-architecture terminology,
this corresponds to space-time scheduling: the spatial decision is processor
assignment, while the temporal decision is execution ordering.

The scheduler operates at sub-job granularity and uses dynamic priorities to decide the execution ordering. The
priority of a ready sub-job may change over time as a function of
deadlines, laxities, processor availability, queue state, graph progress, and
online scheduling decisions. HRT sub-jobs execute non-preemptively. If an
executing SRT sub-job must be interrupted to preserve HRT correctness, its
partial execution is discarded and any later execution is treated as a new
dispatch. We model this event as cancellation rather than preemption with
resumption, which reflects accelerator architectures where saving and restoring
the internal execution state may be unsupported or prohibitively expensive.

A schedule assigns each ready sub-job $u$ to a processor $m_u$ and to an
execution interval $[s_u,f_u)$. A schedule is valid if it satisfies:
(i) eligibility, $m_u\in\mathcal M_u$ for every executed \mbox{sub-job} $u$;
(ii) processor exclusivity, no processor executes more than one sub-job at any
time; (iii) nonparallel execution, no sub-job executes on more than one processor
at any time; and (iv) precedence preservation, $(u,v)\in E_J$ implies
$s_v\ge f_u$.

The scheduling objective is lexicographic. First, all admitted HRT jobs must
meet their deadlines under all admissible runtime executions. Second, subject to
this hard constraint, the scheduler should improve SRT responsiveness,
soft-deadline hit ratios, client-level SLO satisfaction, and accelerator
utilization by reducing idle intervals across heterogeneous
accelerator resources.

\section{Framework Overview}
\label{sec:framework}

The proposed framework combines adaptive online scheduling with strict
hard real-time (HRT) guarantees for heterogeneous multi-accelerator
systems executing both HRT and soft real-time (SRT) DNN inference
workloads.

The key idea is to separate \emph{safety} from \emph{performance
optimization}. Runtime performance is optimized by an online scheduler,
while HRT correctness is enforced independently through a deterministic
runtime certification mechanism.

Figure~\ref{fig:overview} illustrates the overall framework.
Offline, the system constructs a feasible HRT-only reference schedule
under worst-case execution assumptions. At runtime, an online scheduling
policy observes the current system state and generates candidate
schedule prefixes that jointly schedule HRT and SRT sub-jobs. Before
execution, each candidate prefix is validated by a runtime safety check.
If the prefix is certified safe, it is executed; otherwise, execution
falls back to the corresponding portion of the certified HRT reference
schedule.
As a consequence, HRT correctness does not depend on the online
scheduler itself. The scheduler may improve SRT responsiveness,
fairness, and accelerator utilization, but unsafe runtime decisions are
filtered by the safety layer before execution.

Although this work employs a reinforcement learning (RL) scheduler, the
framework itself is policy-agnostic. Any online scheduling policy 
capable of generating runtime candidate schedules can be enclosed within
the same safety envelope.

\begin{figure}
    \centering
    \includegraphics[width=0.9\linewidth]{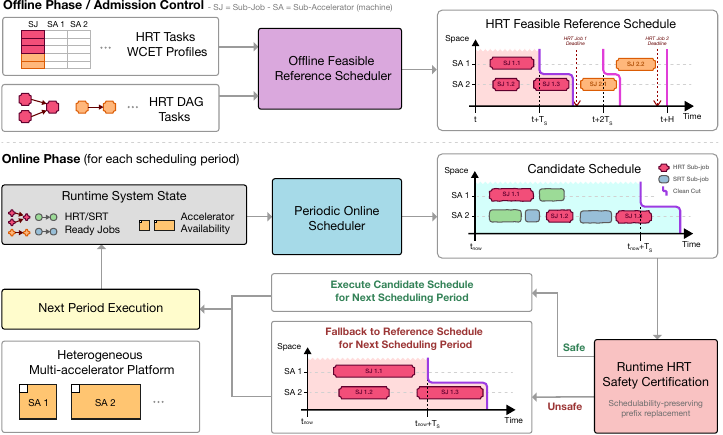}
    \caption{Framework overview.}
    \label{fig:overview}
\end{figure}

\subsection{Reference Schedule and Admission Control}
\label{subsec:reference_schedule}

Hard real-time correctness is derived from a feasible offline reference
schedule constructed considering only the admitted HRT workload. During
system initialization, the framework computes a WCET-based HRT-only
schedule
$S^{\mathrm{ref}}$
over a finite analysis horizon. For periodic tasks, the natural horizon
is the hyperperiod
$
H=\operatorname{lcm}(T_1,\ldots,T_n).
$

For each HRT sub-job $u$, the reference schedule records its processor
assignment and execution interval
$
(s_u^{\mathrm{ref}},f_u^{\mathrm{ref}},m_u^{\mathrm{ref}}).
$
\\The reference schedule is feasible if all HRT jobs meet their deadlines:
$
f_{J_{i,k}^{\mathrm H}}^{\mathrm{ref}}
\le
d_{i,k}^{\mathrm H}$

$\forall J_{i,k}^{\mathrm H}.
$ At runtime, additional HRT tasks may request admission into the system.
When a new HRT task is submitted, the framework performs admission
control by attempting to construct a new feasible HRT-only reference
schedule including both the already admitted HRT tasks and the new
candidate task. If such a feasible schedule exists, the task is admitted
and the newly generated schedule replaces the previous reference
schedule. Otherwise, the task is rejected. A reference schedule is constructed offline, hence, its construction complexity does not affect runtime scheduling, and the framework is independent of the construction method.
SRT clients are not admitted through the HRT feasibility test; a SRT client only
registers its target SLO achievement rate.
To improve the availability of contiguous slack regions, the reference
schedule may optionally be transformed through a uniform forward shift
of HRT executions while preserving all deadlines, so that at least one sub-job starts at its \textit{latest start time} (LST).

\subsection{Runtime Scheduling}
\label{subsec:scheduler_activation}

The online scheduler is activated periodically with period $T_S$. At
activation time
$
t_k = kT_S,
$
the scheduler observes the current runtime state, including ready HRT
and SRT sub-jobs, processor availability, queue occupancy, graph
progress, and remaining execution tails inherited from previous
scheduler windows.

In
heterogeneous DNN inference systems, allocation and priority decisions~\cite{davis2011survey} are tightly
coupled because execution times depend on the selected accelerator.

Based on the observed runtime state, the scheduler generates a candidate
schedule prefix for a bounded future interval. Unlike the offline
reference schedule, the runtime scheduler is fully adaptive and may
execute HRT sub-jobs earlier than in the reference schedule, dynamically
reorder HRT and SRT work, select different eligible accelerators, leave
processors temporarily idle, or opportunistically exploit runtime slack
for SRT execution.

The scheduler activation period $T_S$ determines the tradeoff between
runtime adaptability and scheduling overhead. Smaller activation periods
improve responsiveness to newly arrived SRT jobs but increase runtime
scheduler overhead.

Let $\tau_{\mathrm{sch}}$ denote an upper bound on the total runtime
overhead including schedule generation, runtime certification,
dispatching, and cancellation handling. This overhead is included in
the certified timing model.

\subsection{Runtime Safety Certification}
\label{subsec:acceptance_problem}

The online scheduler is not trusted for HRT correctness. Consequently,
every runtime-generated candidate prefix must be validated before
execution.

A natural solution would be to directly test whether a candidate prefix
preserves future HRT schedulability. However, exact runtime feasibility
analysis is computationally impractical for heterogeneous DAG-based
systems with precedence constraints and non-preemptive execution.
Instead, the framework employs a bounded-time sufficient acceptance
test derived from the offline reference schedule. Intuitively, the
runtime certification mechanism verifies whether the candidate prefix
preserves enough certified progress to guarantee that all unfinished HRT
work can still complete feasibly after the checked cut.

If the candidate prefix is certified safe, execution proceeds according
to the online schedule. Otherwise, execution falls back to the
corresponding portion of the offline HRT reference schedule and all SRT
work contained in the rejected prefix is discarded.
The acceptance test therefore acts as a deterministic safety envelope
surrounding the online scheduler. Unsafe runtime decisions may reduce
performance, but they cannot compromise HRT correctness.

Section~\ref{sec:theory} formalizes the runtime certification mechanism
through the notion of schedulability-preserving prefix replacement.

\section{Runtime Safety Certification}
\label{sec:theory}

This section formalizes the runtime certification mechanism used to
guarantee hard real-time correctness independently of the online
scheduler.

We first recall the intuition behind transfer schedulability on
uniprocessors and explain why its classical scalar slack formulation
does not directly extend to heterogeneous precedence-constrained
systems. We then introduce a schedulability-preserving prefix
replacement criterion suitable for heterogeneous multi-accelerator DNN
scheduling.

\subsection{From Transfer Schedulability to Prefix Certification}
\label{sec:transfer}

Transfer schedulability~\cite{willemsen2025transfer} provides a way to
preserve schedulability by comparing an online execution with a feasible
reference schedule. In its classical uniprocessor form, the online scheduler
monitors whether the work that must match the reference schedule has become
locally urgent. When a slackless interval is detected, the scheduler must give
priority to the corresponding locally urgent work; otherwise, schedulability
may no longer be transferable from the reference schedule.

We use this idea only as a design principle. We do not attempt to extend the
uniprocessor slackless test directly to heterogeneous multiprocessor DAG
scheduling. In our setting, feasibility depends jointly on machine assignment,
non-preemptive execution, precedence constraints, and future resource
contention across unrelated accelerators. Therefore, a scalar slack quantity is
not sufficient to characterize whether the online state can still match a
feasible continuation.

To see why, consider even the simpler case of a set $\mathcal J$ of independent, non-preemptive,
non-migratory jobs on $M$ identical processors over an interval of length $L$.
Let
\[
    W = \sum_{J\in\mathcal J} C_J
\]
be the total remaining work, being $C_J$ the length of a job $J$. The condition $W > ML$ proves infeasibility, but
the converse is false: $W\le ML$ does not imply that the jobs can be packed on
the individual processors. For example, with two processors, $L=10$, and three
jobs of length $C_J=6$, the total work is $W=3 C_J=18<20$, but no feasible assignment
exists. The corresponding decision problem contains $P||C_{\max}\le L$ as a
special case, which is NP-hard by reduction from \textsc{Partition}.

The setting considered in this paper is strictly more general: processors are
unrelated, since a sub-job $u$ has execution time $C_{u,m}$ on machine $m$, and
HRT jobs are DAGs whose vertices are non-preemptive sub-jobs. Consequently, the
runtime safety layer deliberately uses a bounded-time sufficient test. If the
test accepts, HRT schedulability is preserved; if it rejects, the system falls
back to a certified continuation.

\subsection{Schedulability-Preserving Prefix Replacement}
\label{sec:prefix_replacement}

Let $\mathcal J^{\mathrm H}$ denote the set of generated HRT DAG jobs under
consideration, and let
\[
    \mathcal V^{\mathrm H} =
    \bigcup_{J\in\mathcal J^{\mathrm H}} V_J
\]
be the set of their HRT sub-jobs. Let $S^{ref}$ be a feasible WCET reference
schedule for $\mathcal J^{\mathrm H}$, and let $S^{cand}$ be a candidate runtime
schedule. For a sub-job $u$, we write $m_S(u)$, $s_S(u)$, and $f_S(u)$ for its
machine, start time, and finish time in schedule $S$.

\begin{definition}[Machine-wise cut]
A machine-wise cut is a vector
\[
    T^\star = (t_1,\ldots,t_M),
\]
where $t_m$ is a time instant referred to machine $m$. We refer to $t_m$ as the local cut time on machine $m$.
\end{definition}

\begin{definition}[HRT Schedule Prefix]
Given a schedule $S$ and a machine-wise cut
$T^\star$, the HRT schedule prefix induced by $T^\star$ is
\[
    P_S(T^\star)
    =
    \left\{
    u\in\mathcal V^{\mathrm H}
    \;\middle|\;
    f_S(u) \le t_{m_S(u)}
    \right\}.
\]
That is, the prefix contains all HRT sub-jobs that complete no later than the local cut time of the machine to which they are assigned.

For a job $J$, we denote by
\[
    P_S^J(T^\star)
    =
    P_S(T^\star)\cap V_J
\]
the portion of its HRT sub-jobs belonging to the prefix.
\label{def:scheduleprefix}
\end{definition}

\begin{definition}[Clean cut]
A machine-wise cut $T^\star$ is clean for schedule $S$ if no HRT sub-job crosses
the cut on its assigned machine, i.e.,
\[
    \nexists u\in\mathcal V^{\mathrm H}
    \quad\text{s.t.}\quad
    s_S(u) < t_{m_S(u)} < f_S(u).
\]
SRT sub-jobs are not part of the HRT correctness argument; if
an SRT sub-job crosses any certified cut, it is canceled at the cut
and its partial execution is discarded.
\end{definition}

\begin{definition}[Closed and open HRT jobs]
An HRT job $J$ is closed in schedule $S$ at cut $T^\star$ if
\[
    V_J \subseteq P_S^J(T^\star).
\]
Otherwise, $J$ is open at $T^\star$.
\end{definition}

\begin{definition}[Open-interface sub-jobs]
For a schedule $S$, an HRT sub-job $u\in P_S^J(T^\star)$ is an open-interface
sub-job of job $J$ if it has at least one successor of the same job outside the
prefix:
\[
    \exists v\in V_J\setminus P_S^J(T^\star)
    \quad\text{s.t.}\quad
    (u,v)\in E_J.
\]
We denote the set of such sub-jobs by $I_S^J(T^\star)$.
\end{definition}

\begin{figure}
    \centering
    \includegraphics[width=\linewidth]{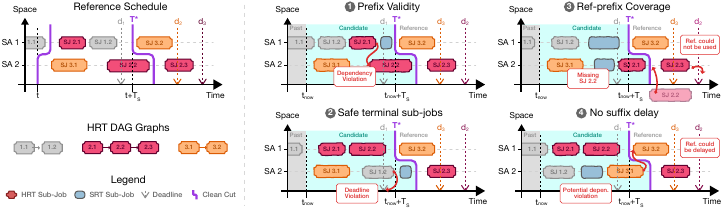}
    \caption{Counterexamples illustrating the necessity of the four conditions in Definition~\ref{def:safeprefixreplacement} required by Theorem~\ref{th:safeprefixreplacement}. Each panel violates one condition while the remaining conditions are satisfied, showing that omitting any single condition may invalidate the schedulability-preserving prefix replacement guarantee.}
    \label{fig:counterexample}
\end{figure}

\begin{definition}[Safe prefix replacement]
\label{def:safeprefixreplacement}
Let $S^{ref}$ be a feasible HRT reference schedule and let $S^{cand}$ be a valid
candidate schedule. We say that $S^{cand}$ is a safe prefix replacement of $S^{ref}$ at
cut $T^\star$ if the following conditions hold.

\begin{enumerate}
\item \textbf{Candidate-prefix validity.}
The HRT prefix $P_{S^{cand}}(T^\star)$ satisfies eligibility, machine
exclusivity, non-parallel execution, and precedence closure: for every
edge $(u,v)\in E_J$, if $v\in P^J_{S^{cand}}(T^\star)$, then
$u\in P^J_{S^{cand}}(T^\star)$ and $f^{cand}_u\le s^{cand}_v$.

   \item \textbf{Completed terminal HRT sub-jobs are already safe.}
    For every HRT job $J$, every terminal sub-job completed
    in the candidate prefix finishes no later than the
    absolute deadline of its job:
\[
    \forall u \in P_{S^{cand}}^J(T^\star),
    \qquad
    \bigl(\nexists v \in V_J : (u,v)\in E_J\bigr)
    \Longrightarrow
    f_u^{cand}\le d_J .
\]
    Here, a terminal sub-job has no outgoing edge in the
    complete DAG of its job. This condition applies to
    both open and closed jobs.

    \item \textbf{Candidate-prefix covers reference-prefix.}
    For every HRT job $J$ in $S^{ref}$, all sub-jobs of
    $J$ that the reference schedule has completed before the cut $T^\star$ are also
    completed in the candidate prefix:
    \[
        P_{S^{ref}}^J(T^\star)
        \subseteq
        P_{S^{cand}}^J(T^\star).
    \]

\item \textbf{No delay of reference-suffix releases.}
For every HRT job $J$ that is open in $S^{cand}$ at $T^\star$, if a completed sub-job $u$ has an unfinished successor $v$, then $u$ must complete no later than the start time of $v$ in the reference schedule. Formally,
for every edge $(u,v)\in E_J$ such that
\[
    u\in P_{S^{cand}}^J(T^\star)
    \qquad\text{and}\qquad
    v\notin P_{S^{cand}}^J(T^\star),
\]
we require
\[
    f^{cand}_u \le s^{ref}_v .
\]
This prevents a sub-job completed by the candidate prefix from
delaying any unfinished successor when execution resumes according to the
reference suffix.

\noindent\emph{Remark.}
A simpler but more conservative runtime implementation can replace the
successor-start comparison above with the stronger sufficient check
$f^{cand}_u \le f^{ref}_u$ for open-interface sub-jobs, since the validity of the
reference schedule implies $f^{ref}_u \le s^{ref}_v$ for every edge~$(u,v)$. 
\end{enumerate}

\end{definition}

\begin{theorem}[Schedulability-preserving prefix replacement]
\label{th:safeprefixreplacement}
Let $S^{ref}$ be a feasible WCET reference schedule for the HRT jobs
$\mathcal J^{\mathrm H}$, and let $S^{cand}$ be a candidate schedule. Assume that
$T^\star$ is clean for both $S^{ref}$ and $S^{cand}$. If $S^{cand}$ is a safe prefix
replacement of $S^{ref}$ at $T^\star$, then replacing the HRT prefix of $S^{ref}$ with
the HRT prefix of $S^{cand}$ preserves HRT schedulability.
\end{theorem}

\begin{proof}
Keep the candidate prefix and execute every remaining HRT sub-job
in its reference interval on its reference machine, skipping all
sub-jobs already completed in the candidate prefix.

By Condition~3, every remaining sub-job lies outside the reference
prefix. Since $T^\star$ is clean for $S^{ref}$, its reference execution
starts no earlier than the local cut on its reference machine.
On each machine, candidate-prefix executions finish no later than
the local cut, so they cannot overlap the retained reference
executions. Together with Condition~1 and reference validity,
this preserves eligibility, machine exclusivity, and non-parallel
execution.

Precedence constraints within the candidate prefix hold by
Condition~1, while those between remaining sub-jobs hold because
their reference intervals are unchanged. For an edge $(u,v)$
from the candidate prefix to a remaining sub-job, Condition~4
ensures $f_u^{cand}\le s_v^{ref}$. Conversely, an edge from a
remaining sub-job into the candidate prefix is excluded by
the precedence closure in Condition~1. Thus all precedence
constraints are preserved.

Finally, consider any HRT job $J$. Every terminal sub-job in its
candidate prefix meets $d_J$ by Condition~2, whether $J$ is open
or closed. Every remaining terminal sub-job meets $d_J$ because
it executes in its feasible reference interval. Since every
non-terminal vertex of a finite DAG has a path to a terminal
vertex, precedence preservation implies that all sub-jobs of
$J$ finish by $d_J$. Hence all HRT jobs meet their deadlines.
\end{proof}

\section{Online Scheduler}
\label{sec:scheduler}

This section describes how the proposed framework applies the
schedulability-preserving prefix replacement theory of
Section~\ref{sec:theory} at runtime.

The framework employs an online scheduler that periodically generates
candidate schedule prefixes for the currently available HRT and SRT
sub-jobs. The scheduler itself is not trusted for hard real-time
correctness. Consequently, every runtime-generated candidate prefix is
validated through the runtime certification mechanism derived from
Theorem~\ref{th:safeprefixreplacement} before execution.

If the candidate prefix satisfies the certification conditions, it is
committed for execution. Otherwise, the system falls back to the
corresponding portion of the certified HRT reference schedule and any
SRT work contained in the rejected prefix is discarded.

Although the framework is scheduler-agnostic, this work employs a
reinforcement learning (RL)-based scheduler that jointly optimizes
accelerator assignment and execution ordering to improve SRT QoS and
accelerator utilization while remaining enclosed within the deterministic
HRT safety envelope.

Section~\ref{sec:check} describes the runtime HRT certification
mechanism, Section~\ref{sec:rl} briefly reviews the RL background used
in this work, and Section~\ref{sec:policy} presents the proposed RL
scheduler policy.

\subsection{Runtime HRT Certification}
\label{sec:check}

\begin{figure}
    \centering
    \includegraphics[width=\linewidth]{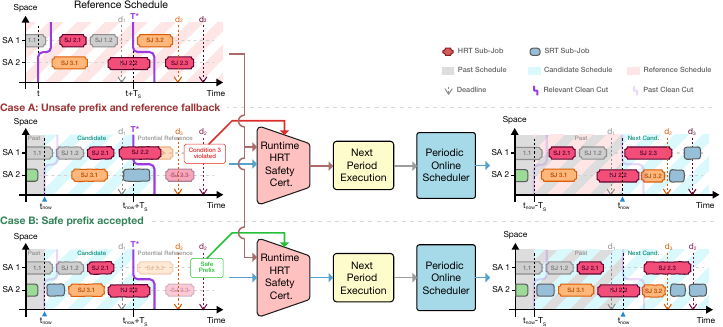}
    \caption{Runtime HRT certification on two sub-accelerators. The candidate prefix is checked against the feasible reference schedule at the machine-wise clean cut $T^\star$ (purple). In Case~A, sub-job SJ~2.2 belongs to the reference prefix but is not completed within the candidate prefix, violating the reference-prefix coverage condition (Condition~3 of Definition~\ref{def:safeprefixreplacement}). The candidate is rejected, and execution falls back to the reference schedule, interrupting SRT execution when necessary. In Case~B, the candidate satisfies the certification conditions and is accepted. After execution, the periodic online scheduler generates the next candidate prefix.}
    \label{fig:certification}
\end{figure}

The runtime safety layer certifies finite HRT prefixes against the WCET
reference schedule. Let $S^{\mathrm{ref},\infty}$ denote the infinite periodic
extension of the feasible reference schedule. All comparisons below are made on
absolute time in this infinite extension; the modulo-$H$ representation is used
only for table lookup in the implementation.

At scheduler activation time $t_k$, the online scheduler observes the current
state and proposes a candidate prefix. Let
$
    h_k = t_k + T_S + \tau_{\mathrm{sch}}
$
be the nominal certification horizon, including the upper bound on runtime
scheduling overhead. For each machine $m$, the safety layer defines the cut
component $t_m^\star$ as the first time not earlier than $h_k$ at which no HRT
sub-job of $S^{\mathrm{ref},\infty}$ is executing on machine $m$. Equivalently,
if an HRT sub-job is executing on machine $m$ at $h_k$, then $t_m^\star$ is the
finish time of that sub-job in the reference schedule; otherwise,
$t_m^\star=h_k$.

The resulting machine-wise cut
$
    T^\star=(t_1^\star,\ldots,t_M^\star)
$
is clean for the reference schedule by construction. The online scheduler must
provide a candidate HRT prefix up to the same machine-wise cut. The candidate is
rejected unless $T^\star$ is also clean for the candidate schedule and the
safe-prefix-replacement conditions of
Definition~\ref{def:safeprefixreplacement} hold.

If the test \emph{accepts}, the candidate prefix is committed up to $T^\star$. Importantly, the certification conditions in Definition~\ref{def:safeprefixreplacement} are expressed in terms of start and finish times and do not require HRT sub-jobs to execute for their full WCET; hence, candidate prefixes may exploit earlier completions as long as the certification conditions remain satisfied.

If the
test \emph{rejects}, all SRT work in the rejected prefix is discarded and the system
executes the corresponding certified reference prefix. In this case, execution conservatively follows the WCET-based reference schedule: the reference slot of each HRT sub-job remains reserved for its full duration even if the sub-job completes earlier, without reclaiming the resulting slack.

Since after each accepted or fallback prefix the remaining HRT work
admits the corresponding reference suffix as a feasible continuation,
induction over scheduler activations guarantees HRT schedulability
for the whole runtime execution.

The nominal parameter $T_S$ controls how frequently the scheduler attempts to
create new candidate prefixes. 

\subsection{Background on Reinforcement Learning}
\label{sec:rl}

RL is a machine learning paradigm in which an
agent learns a decision policy through repeated interaction with an
environment. At each decision step $t$, the agent observes the current
state $s_t$, selects an action $a_t$ according to a policy
$\pi(a_t|s_t)$, and receives a scalar reward $r_t$ reflecting the
quality of the selected action. The environment then transitions to a
new state $s_{t+1}$~\cite{sutton1998reinforcement}.

The objective of RL is to learn a policy maximizing the expected
discounted cumulative reward:
$
J(\pi)
=
\mathbb{E}_{\pi}
\left[
\sum_{t=0}^{\infty}
\eta^t r_t
\right],
$
where $\eta \in [0,1)$ is the reward discount factor.

In actor--critic methods, the policy is represented by an
\emph{actor} network, while a separate \emph{critic} network estimates
the expected long-term utility of actions. 
More specifically, the critic
approximates the action-value function
$
Q^\pi(s,a)
=
\mathbb{E}_{\pi}
\left[
\sum_{k=0}^{\infty}
\eta^k r_{t+k}
\;\middle|\;
s_t=s,\ a_t=a
\right].
$

This work employs Deep Deterministic Policy Gradient
(DDPG)~\cite{lillicrap2015continuous}, an off-policy actor--critic
algorithm designed for continuous action spaces. DDPG represents both
actor and critic through deep neural networks and updates them during the \emph{training} phase, using
transitions collected during runtime interaction with the environment.
Training samples are stored in a replay buffer as tuples
$
(s_t,a_t,r_t,s_{t+1}),
$
which are later sampled uniformly to reduce temporal correlation and
improve training stability. The critic is trained to minimize the
Bellman error:
$
L(\phi)
=
\mathbb{E}
\left[
\left(
Q_\phi(s_t,a_t)
-
y_t
\right)^2
\right],
$
where the target value is
$
y_t
=
r_t
+
\eta
Q_{\phi^-}
\left(
s_{t+1},
\pi_{\theta^-}(s_{t+1})
\right),
$
and $\phi^-$ and $\theta^-$ denote slowly updated target-network
parameters.
The actor parameters $\theta$ are updated through deterministic policy
gradients:
$
\nabla_\theta J(\theta)
=
\mathbb{E}
\left[
\nabla_a Q_\phi(s,a)
\big|_{a=\pi_\theta(s)}
\nabla_\theta \pi_\theta(s)
\right].
$
After the \emph{training} phase, the critic is no longer used during \emph{deployment}, and the actor alone serves as the scheduling policy.
DDPG is particularly suitable for scheduling problems because it
naturally supports high-dimensional continuous action spaces, allowing
the policy to jointly encode scheduling priorities and accelerator
selection decisions.

\subsection{RL-based Scheduler Policy}
\label{sec:policy}

The online scheduler considered in this work is based on the DDPG~\cite{lillicrap2015continuous} algorithm described in Section~\ref{sec:rl}. We considered alternative RL formulations, including discrete-action methods such as Q-learning. However, directly representing the joint sub-job ordering and accelerator-assignment decisions as discrete actions would result in a combinatorial action space that grows rapidly with the number of ready sub-jobs and available accelerators. Even other formulations, such as multi-armed bandit formulations~\cite{bandit1,bandit2}, are also not well suited to this problem because scheduling decisions are inherently sequential and affect future system states; capturing such state transitions and delayed consequences therefore requires a sequential decision-making formulation. We adopt DDPG and encode the scheduling decisions through continuous priority values, which are subsequently decoded into execution priorities and accelerator assignments.

Importantly, priorities are an internal mechanism used by the online policy to encode and construct a candidate schedule. Once the candidate schedule has been generated, its validity and certification depend only on the resulting machine assignments, execution intervals, precedence constraints, and not on the priority values that produced it. Therefore, priorities do not appear in Definition~\ref{def:safeprefixreplacement}.

We consider the state as the ready sub-jobs and the current accelerator occupancy. At each scheduler activation, the policy observes the current runtime state and generates a candidate schedule by jointly selecting sub-job priorities and accelerator assignments. This formulation allows the scheduler to adapt to time-varying workload arrivals and heterogeneous accelerator availability while accounting for precedence constraints, resource contention, and execution tails.

\begin{figure}
    \centering
    \includegraphics[width=\linewidth]{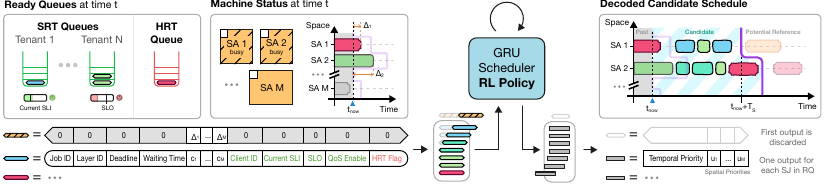}
    \caption{State and action encoding of the RL scheduler. The input sequence is prepended with a vector encoding the residual occupancy of each accelerator: an idle accelerator is represented by zero (e.g., $\Delta_M=0$), whereas a busy accelerator is represented by its remaining busy time (e.g., $\Delta_1>0$). Each subsequent element encodes a ready HRT or SRT sub-job. For each input element, the policy outputs a temporal priority and $M$ accelerator-selection scores. The first output, corresponding to the accelerator-occupancy state vector rather than to a sub-job, is discarded.}
    \label{fig:encoding}
\end{figure}

As shown in Figure~\ref{fig:encoding}, the runtime state is encoded as a sequence of fixed-length feature
vectors, one for each ready sub-job in ready queue $RQ_t$, naturally supporting a variable
number of schedulable sub-jobs. The feature representation includes
model and layer identifiers, relative deadlines, waiting times,
profiled execution latencies for each
accelerator, the identifier of the client generating the request, the
current client-level deadline-hit ratio
$H_k^{\mathrm S}$, the corresponding target SLO achievement rate
$\hat H_k^{\mathrm S}$, and two boolean flags identifying HRT and
QoS-aware sub-jobs. Each feature vector therefore has length
$9+M$, where $M$ is the number of accelerators.

The input sequence is preceded by an additional $9+M$-dimensional all-zero vector except $M$ entries, which show the residual busy time of each accelerator inherited from previously
dispatched execution tails, allowing the policy to reason about
near-future accelerator availability and congestion. Ready sub-jobs are then provided in deadline order.

For each ready sub-job, the actor network emits an action encoded as an $(M+1)$-dimensional vector, as follows:
(i) a scalar temporal priority and
(ii) an $M$-dimensional spatial priority vector over the accelerators.
The maximum spatial priority among the eligible accelerators in $\mathcal{M}_v$ determines the selected accelerator.

The continuous policy outputs are decoded into a discrete schedule by
dispatching, whenever an accelerator becomes available, the
highest-priority eligible sub-job onto the accelerator selected by the
spatial priority vector.

Both actor and critic are implemented as recurrent neural networks to
support variable-sized ready queues. In our implementation, we employ
GRU layers for improved computational efficiency. The recurrent encoder
captures sequence-level scheduling context such as deadline structure,
accelerator contention, and queue composition. The critic estimates the
action-value function $Q_\phi(s,a)$ from the recurrent hidden state.

Training follows standard off-policy DDPG methodology. Runtime
transitions
$
(s_t,a_t,r_t,s_{t+1})
$
are stored in a replay buffer and sampled uniformly to reduce temporal
correlation. Target networks and soft parameter updates are used to
stabilize training~\cite{lillicrap2015continuous}. Exploration noise is
applied only during training.

We designed a reward function to prioritize deadline satisfaction and, among feasible outcomes, improve client-level QoS balance and accelerator utilization. The relative importance of the different reward components is controlled by tunable parameters, which are set before training and can be selected empirically, e.g., through grid search, according to the desired trade-off among the optimization objectives.
The reward is decomposed into four components:
\begin{align}
r_t
&=
R^{\mathrm{deadline}}_t
+
R^{\mathrm{slack}}_t
+
R^{\mathrm{QoS}}_t
+
R^{\mathrm{HRT}}_t.
\end{align}

The first three reward components, \textit{i.e.}, $R^{\mathrm{deadline}}_t$, $R^{\mathrm{slack}}_t$, and $R^{\mathrm{QoS}}_t$, pertain to the SRT workload, whereas $R^{\mathrm{HRT}}_t$ pertains to the HRT workload.

The deadline satisfaction component, which penalizes deadline misses, is:
\begin{align}
R^{\mathrm{deadline}}_t
&=
\sum_{l \in RQ_t^{SRT}}
\Delta_{t,l}
A_{t,l},
\end{align}
where
\begin{align}
A_{t,l}
&=
\begin{cases}
\alpha, & f_l \le d_l \\
-\beta, & \text{otherwise}.
\end{cases}
\end{align}

where $\alpha$ and $\beta$ are tunable reward hyperparameters.

The normalized completion-time improvement component, which favors earlier completions, is:
\begin{align}
R^{\mathrm{slack}}_t
&=
\gamma
\sum_{l \in RQ_t^{SRT}}
\Delta_{t,l}
\frac{d_l - f_l}{d_l-a_l}.
\end{align}

The client-level QoS balancing component, which encourages client-level SLO achievement rates to remain close to their target values, is:
\begin{align}
R^{\mathrm{QoS}}_t
&=
-\lambda
\sum_{l \in RQ_t^{SRT}}
\Delta_{t,l}
\left|
H_{\kappa_l}^{\mathrm S}
-
\hat H_{\kappa_l}^{\mathrm S}
\right|,
\end{align}
where $H_{\kappa_l}^{\mathrm S}$ is the current deadline-hit ratio of
the client generating sub-job $l$, and
$\hat H_{\kappa_l}^{\mathrm S}$ is the corresponding target SLO
achievement rate defined in
Section~\ref{subsec:client-qos}; $\lambda$ is a tunable reward hyperparameter.

Finally, HRT deadline misses are strongly penalized through:
\begin{align}
R^{\mathrm{HRT}}_t
&=
-\mu
N^{\mathrm{HRT}}_{\mathrm{miss}}(t),
\end{align}
where $N^{\mathrm{HRT}}_{\mathrm{miss}}(t)$ is the number of HRT
deadline misses observed during the scheduling window and
$\mu \gg 1$ is a large penalty coefficient.

The over-the-horizon weighting factor is:
\begin{align}
\Delta_{t,l}
&=
\begin{cases}
1, & f_l < t + T_S \\
\delta, & \text{otherwise}.
\end{cases}
\end{align}

During training, any HRT deadline miss produces a catastrophic negative reward through $R^{\mathrm{HRT}}_t$. Unlike for the SRT workload, where $R^{\mathrm{slack}}_t$ encourages earlier completion, the HRT reward does not incentivize completion earlier than necessary. During deployment, exploration is
disabled and candidate schedules are executed only if accepted by the
schedulability-preserving runtime certification mechanism described in
Section~\ref{sec:theory}. Consequently, the RL policy affects only
QoS-aware optimization and candidate schedule quality, while HRT
correctness is entirely enforced by the deterministic safety layer.
Regarding the reference schedule construction, in the experiments we have employed several techniques. The main two are: EDF ordering with earliest-completion accelerator selection, and a recurrent heuristic inspired by~\cite{RELMAS1} favouring larger deadline slack. Both methods are followed by deterministic validation of precedences and HRT deadlines.

\section{Experimental Results}

\subsection{Experimental Setup}
We evaluate the proposed framework on a heterogeneous multi-accelerator
platform with six DNN accelerators, following the established heterogeneous setup of~\cite{RELMAS1,RELMAS2}, enabling the comparison with prior scheduling techniques.
In architectural terms, three accelerators implement a row-stationary dataflow and an organization
close to Eyeriss~\cite{chen2019eyeriss}; the other three implement a
weight-stationary dataflow and an organization close to Simba~\cite{shao_micro19}.
For each accelerator family, two instances are small and one instance is large.
Hence, the platform exposes both dataflow heterogeneity and size heterogeneity,
and sub-job execution times depend on the selected accelerator.
We consider representative DNN inference workloads from image classification
and object detection. The image-classification set includes AlexNet,
InceptionV3, and ResNet50, while the object-detection set includes YOLO-v2.
Each inference request is represented as a DAG of layer-level sub-jobs, with
precedence constraints induced by the network topology. For every layer and
eligible accelerator, we use profiled execution latencies as the scheduling
costs.
SRT arrivals follow a Google-like workload model derived from large-scale
production traces. Following the workload generator of
Da Costa et al.~\cite{da2018modeling}, inter-arrival times are sampled from a
heavy-tailed Pareto process. Let \(d\) denote the target mean inter-arrival time,
corresponding to the generator's dynamism parameter. For each submitted SRT job
\(i\), we sample
$
    X_i \sim \mathrm{Pareto}(\alpha=4,x_m=1)
$
and set
$
    a_i = a_{i-1} + 3d(X_i-1).
$
Since \(\mathbb{E}[X_i-1]=1/3\), the expected inter-arrival time is
$
    \mathbb{E}[a_i-a_{i-1}] = d .
$
We evaluate two workload intensities: \(d=40\,\mathrm{ms}\), denoting a
moderate-pressure workload, and \(d=10\,\mathrm{ms}\), denoting a high-pressure
workload with four times more frequent SRT arrivals.

Each SRT request is associated with a client \(k\in\mathcal K\). 
For each
client, we measure the client-level SLO achievement rate \(H_k^{\mathrm S}\),
defined as the fraction of that client's SRT jobs that complete by their soft
deadline. Client targets are chosen from
$
    \hat H_k^{\mathrm S}\in\{70\%,80\%,90\%\}.
$
The aggregate SRT hit ratio \(H^{\mathrm S}\) measures system-level SRT
timeliness, while the quantity
$
    \min_{k\in\mathcal K}
    \left(H_k^{\mathrm S}-\hat H_k^{\mathrm S}\right)
$
captures the worst client-level SLO margin. If this value is non-negative, all
clients meet their target SLO achievement rate; negative values indicate that at
least one client is below its target. Therefore, our secondary objective, after
preserving HRT schedulability, is to improve \(H^{\mathrm S}\) while avoiding
negative client-level SLO margins for each client.

Across all certified experiments, no admitted HRT job misses its
deadline; this only validates the implementation, whereas the HRT
guarantee follows from Theorem~\ref{th:safeprefixreplacement}. The measurement simply
confirms that the implementation follows the certified execution protocol.
Regarding the reward parameters, we employed $\alpha=0.01$ for SRT deadline hits, $\beta=0.11$ for SRT misses, $\mu=30$ for HRT misses, $\lambda=0.1$ for client-level QoS imbalance, and $\delta=0.05$ for outcomes beyond the current scheduling window. These coefficients were selected through a grid search over training simulations. For example, $\mu$ was chosen so that HRT misses make the total reward negative despite the other components, while the remaining terms were balanced to prevent any secondary objective from dominating.

\subsection{Baselines}
\label{sec:baselines}

We compare against representative online DNN schedulers, slack-based real-time
mechanisms, and simple dispatch heuristics, all implemented at layer granularity
with the same profiled accelerator latencies. Methods that do not natively
provide HRT guarantees in our heterogeneous DAG-based model are evaluated as
HRT-safe wrapped variants, obtained by enclosing them in the same
prefix-replacement safety envelope used by our framework.
This also shows that the proposed certification-and-fallback mechanism is not
specific to our RL scheduler. Any online scheduling policy can be enclosed
within the same safety envelope: if the policy is effective at scheduling real-time jobs on heterogeneous platforms, the fallback mechanism can be used
to make its execution HRT-safe.
We include PREMA~\cite{choi2020prema}, adapting its
priority/allocation logic to our heterogeneous layer-level simulator. In particular, our adaptation selects among the highest-aging ready layers according to their expected latency, and assigns them to accelerators considering current availability. We also
include a client-aware RL policy~\cite{RELMAS2} that augments the policy state and reward with client-level SLO information
based on the deadline-hit rate SLI.
We also employ the algorithms proposed by Tang et al.~\cite{tang2011combining} by constructing an HRT schedule
with H-EDF or FF, deriving per-accelerator slack windows, and admitting SRT
layers only when they can be placed by RC, TC, or TCSD without violating release
times, latest-start-time bounds, machine order, or DNN precedence constraints.
We implement SSML~\cite{slackstealing}, originally formulated for a uniprocessor setting, as a conservative scheduler-core adaptation that computes a modified-laEDF slack budget $\sigma$ from the remaining HRT layer workload and admits a non-preemptive SRT layer only if its whole execution fits both the available slack budget and the current scheduling window.
We adapt CD-MSA~\cite{wang2023cdmsa} by applying its EDF ordering, $\sum_{i=1}^{k} C_i \le D_k$ schedulability test, and laxity-aware promotion rule at the DNN-request level, while estimating remaining execution costs from our heterogeneous layer latencies and subsequently mapping the selected layers to accelerators.
We also include SJF, SSJF, FCFS, and LCFS as simple online dispatch
heuristics for ready DNN layers, augmented with a heuristic to select, for each sub-job, the accelerator expected to yield the shortest completion time based on dataflow affinity and current resource utilization.


\subsection{Experimental Results}
\label{sec:experimental_results}

\paragraph{HRT-safe evaluation protocol.}
All methods are evaluated under preserved HRT schedulability: every admitted
HRT job must meet its deadline. 
For baselines that do not provide HRT guarantees in our heterogeneous
DAG-based model, including PREMA~\cite{choi2020prema}, the client-aware RL baseline~\cite{RELMAS2}, CD-MSA~\cite{wang2023cdmsa}, SJF, SSJF, FCFS, and LCFS, we report their certified variants, obtained by applying the
same prefix-replacement safety envelope used by our scheduler. This evaluation
isolates the quality of the online scheduling policy under a common HRT
safety mechanism. We also report an ablation without the safety envelope to
show the HRT losses incurred by the raw policies. In our experiments, we also include the natively certified FF+RC/FF+TC/FF+TCSD~\cite{tang2011combining} algorithms\footnote{We have chosen FF over H-EDF as the periodic task scheduler since it was designed specifically as a heterogeneous task scheduler and it results in the best performance in our experiments.}.

\begin{figure}[!t]
    \centering
    \includegraphics[width=0.89\columnwidth]{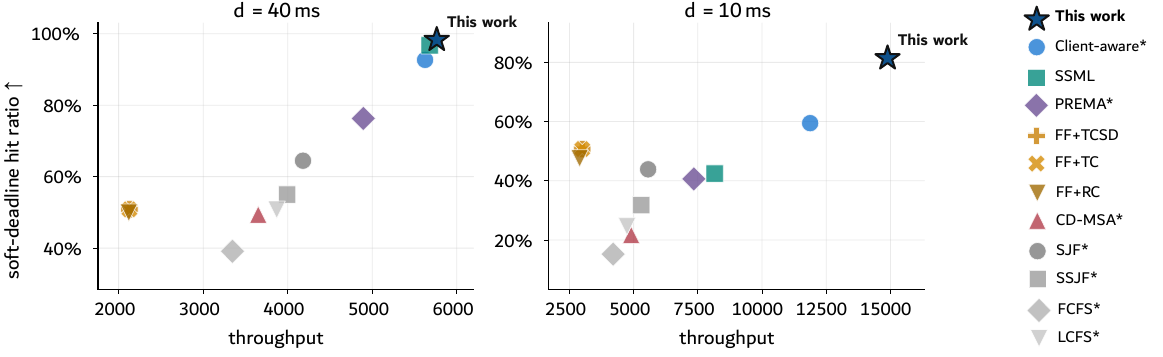}

    \caption{Throughput versus SRT hit ratio \(H^S\). The paired plot compares
    the moderate-pressure workload, 40 ms inter-arrival, with the high-pressure
    workload, 10 ms inter-arrival. (*) denotes a baseline wrapped with the proposed HRT safety envelope; rejected prefixes fall back to \(S^{\mathrm{ref}}\).}
    \label{fig:throughput_vs_srt_pair}
\end{figure}

\begin{figure}[!t]
    \centering
    \includegraphics[width=0.87\columnwidth]{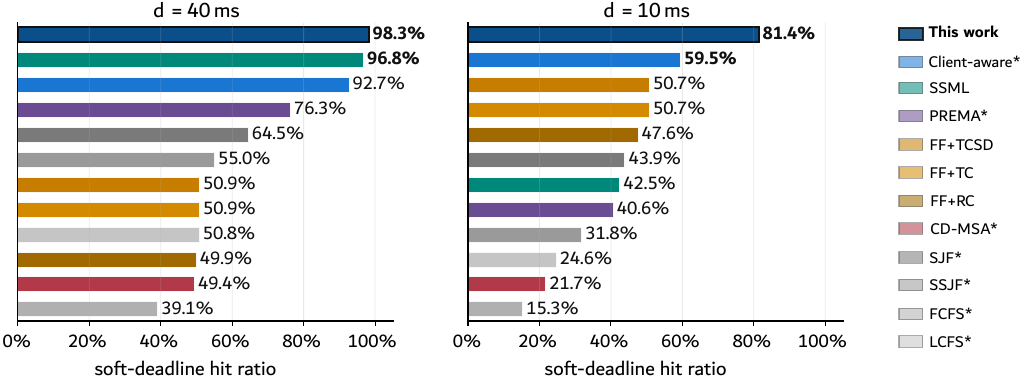}

    \caption{SRT hit ratio \(H^S\) under moderate-pressure and high-pressure
    workloads. At 40 ms, several methods remain competitive; at 10 ms, most baselines sharply lose SRT deadline satisfaction, while the proposed
    scheduler maintains a substantially higher \(H^S\). (*) denotes a baseline wrapped with the proposed HRT safety envelope; rejected prefixes fall back to \(S^{\mathrm{ref}}\).}
    \label{fig:srt_hit_ratio_pair}
\end{figure}

\paragraph{System-level SRT performance.}
At 40 ms inter-arrival, the workload is moderate and the best methods are close:
our scheduler reaches \(H^S=98.29\%\)\footnote{We trained and tested policies under different arrival rates. In this graph, we report the results coming from one policy trained with $d=10ms$ and evaluated at both $d=10ms$ and $d=40ms$. At $d=40ms$, this work achieved a $98.29\%$ SRT hit ratio, compared with $97.43\%$ when employing a policy trained directly at $d=40 ms$, showing effective generalization of the proposed approach.}, followed by SSML at 96.75\% and the client-aware at
92.70\%. The difference appears when SRT arrivals become more frequent. At
10 ms, our scheduler still reaches \(H^S=81.41\%\), while the best baselines, the client-aware and
FF+TC/FF+TCSD, stop at 59.50\% and 50.74\%, respectively. This is a \(+21.91\) percentage-point gain
over the best competing SRT hit ratio. The throughput follows the same trend, as depicted in Figure~\ref{fig:throughput_vs_srt_pair}.

\paragraph{Impact of the reference schedule.}
We conducted a study using 20 distinct feasible HRT reference schedules alternative to the one used in Figure~\ref{fig:srt_hit_ratio_pair} ($d=10ms$), generated via different fine-tuned actor networks and summing seed-controlled Gaussian perturbations to the actor outputs. The results showed $H^S$ from $80.39\%$ to $81.93\%$.

\begin{figure}[!t]
    \centering
    \includegraphics[width=0.89\columnwidth]{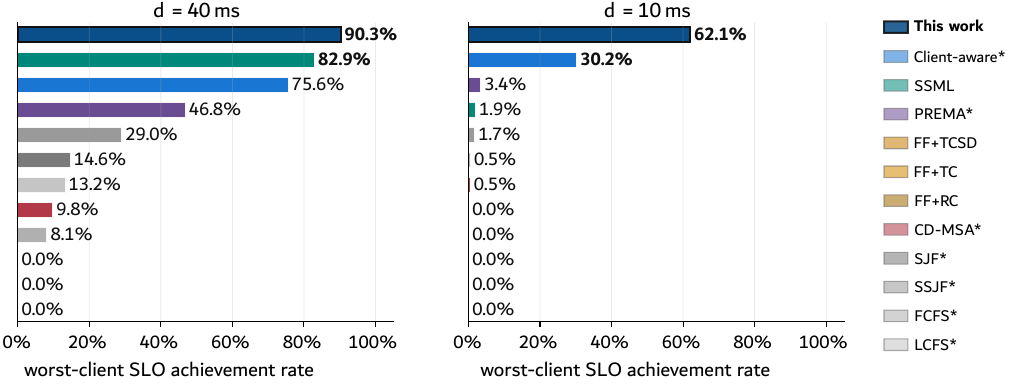}
    \caption{Client-level SRT performance. Worst-client SLO achievement rate,
    \(\min_k H^S_k\), under preserved HRT schedulability. The paired plot
    compares the moderate-pressure workload, 40 ms inter-arrival, with the
    high-pressure workload, 10 ms inter-arrival. At 40 ms, the proposed scheduler
    reaches \(90.3\%\); at 10 ms, it keeps the worst client at \(62.1\%\), while
    the strongest baseline drops substantially. (*) denotes a baseline wrapped with the proposed HRT safety envelope; rejected prefixes fall back to \(S^{\mathrm{ref}}\).}
    \label{fig:worst_tenant_results}
\end{figure}

\paragraph{Worst-client SLO achievement rate.}
The aggregate \(H^S\) alone is not sufficient: a scheduler may complete many SRT
jobs while starving some clients. Figure~\ref{fig:worst_tenant_results} reports
\(\min_k H^S_k\). At 40 ms, our scheduler reaches 90.32\%, followed by SSML at
82.93\% and the client-aware at 75.59\%. At 10 ms, while the client aware drops to 30.19\%, the other baselines collapse at client level: PREMA drops to 3.39\%, SSML to 1.94\%, SSJF to 1.70\%, CD-MSA to 0.48\%, and
several methods reach 0.00\%. Our scheduler keeps the worst client at 62.05\%,
i.e., \(+31.86\) percentage points over the best baseline.

\begin{figure}[!t]
    \centering

    \begin{minipage}[t]{0.40\columnwidth}
        \centering
        \includegraphics[width=\linewidth]{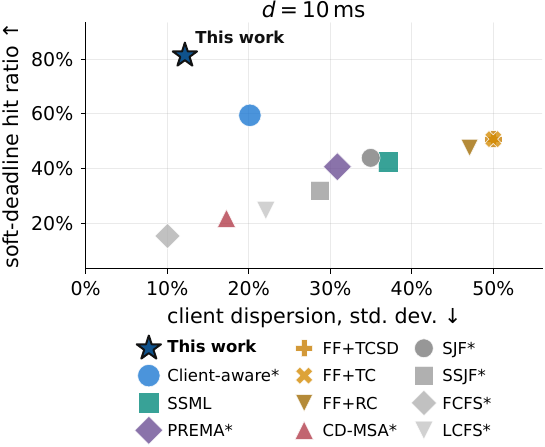}
        
        {\scriptsize (a) \(H^S\) vs. std. dev. of \(H^S_k\).}
    \end{minipage}
    \hfill
    \begin{minipage}[t]{0.29\columnwidth}
        \centering
        \includegraphics[width=\linewidth]{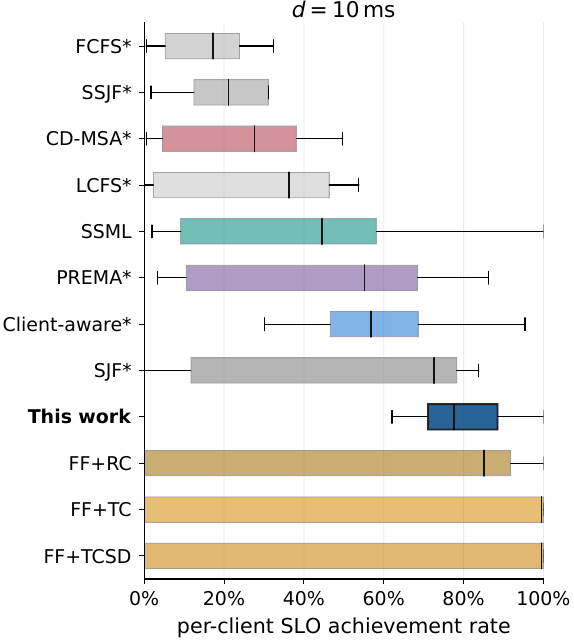}
        
        {\scriptsize (b) Distribution of \(H^S_k\).}
    \end{minipage}
    \hfill
    \begin{minipage}[t]{0.29\columnwidth}
        \centering
        \includegraphics[width=\linewidth]{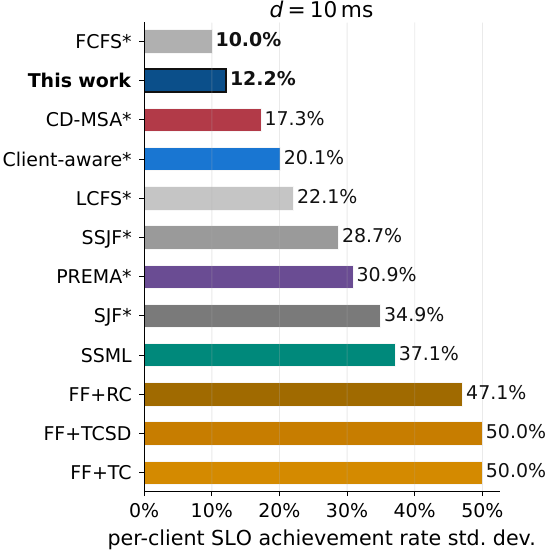}
        
        {\scriptsize (c) Std. dev. of \(H^S_k\).}
    \end{minipage}

    \caption{Client-level behavior at 10 ms inter-arrival. The proposed
    scheduler keeps high SRT hit ratio while controlling the dispersion of
    client-level SLO achievement rates. (*) denotes a baseline wrapped with the
    proposed HRT safety envelope; rejected prefixes fall back to
    \(S^{\mathrm{ref}}\).}
    \label{fig:tenant_results_10ms}
\end{figure}

\paragraph{Client-level distribution at high pressure.}
At 10 ms, our scheduler achieves \(H^S=81.41\%\) with a standard deviation of
client-level rates equal to 12.16\%. FCFS has a slightly lower standard
deviation, 10.03\%, but only \(H^S=15.27\%\), so it is uniformly poor. On the
opposite side, FF+TC and FF+TCSD reach \(H^S=50.74\%\), but their standard
deviation is 49.99\%, with \(\min_k H^S_k=0.00\%\). Their boxplots therefore
show a polarized behavior: some clients are served very well, while others
receive no SRT service. Our scheduler avoids this collapse: the median client
rate is 77.61\%, the interquartile range is 71.05--88.52\%, and the worst client
remains at 62.05\%. Overall, when the workload changes from 40 ms to 10 ms inter-arrival, the
baselines move from acceptable aggregate behavior to poor client-level behavior.
The proposed scheduler remains HRT-safe by construction and keeps high SRT
performance: \(H^S=81.41\%\), high throughput, \(\min_k H^S_k=62.05\%\),
and standard deviation of \((H^S_k)\) equal to $12.16\%$.

\paragraph{Reference-schedule fallback rate.}
Figure~\ref{fig:reference_fallback_rate} reports the
\emph{reference-schedule fallback rate}, i.e., the fraction of scheduler
activations for which the candidate schedule prefix is rejected by the runtime
acceptance test and execution falls back to the certified HRT reference
schedule \(S^{\mathrm{ref}}\).
This metric is important because it quantifies how often the online scheduler,
by itself, fails to produce a certifiable prefix under the HRT safety
constraints.

Algorithms marked with an asterisk are \emph{not} natively HRT-safe in our
model. To enforce HRT correctness and enable a fair comparison under preserved
HRT schedulability, we enclose them within the same deterministic
prefix-replacement safety envelope used by our framework. Therefore, the
asterisk denotes an \emph{HRT-safe wrapped variant} of the original scheduler.
A lower fallback rate indicates that the underlying online scheduler is more
compatible with HRT constraints, since it more frequently generates prefixes
that can be certified safe without reverting to \(S^{\mathrm{ref}}\).

Our scheduler exhibits a very low fallback rate because HRT-awareness is built
directly into the learning process. In particular, HRT jobs are explicitly
identified in the policy input, and HRT deadline violations are assigned a large
negative reward during training. As a result, the learned policy does not
systematically rely on the fallback mechanism; instead, it learns to produce
candidate prefixes that are already close to being certifiable under the HRT
safety test.
\paragraph{Ablation without fallback.}
We also evaluate the schedulers with the prefix-replacement safety envelope
disabled, \ie candidate prefixes are executed directly without fallback to
the certified reference schedule. This experiment is diagnostic only: HRT
correctness is no longer guaranteed. As shown in
Figure~\ref{fig:no_safety_envelope_ablation}, most baselines incur substantial
HRT deadline losses. Our policy reaches \(99.8\%\) HRT deadline satisfaction
and \(81.4\%\) SRT deadline satisfaction, because HRT/SRT job type is part of
the policy input and HRT deadline misses are strongly penalized during
training. However, \(99.8\%\) is still not an HRT guarantee. This confirms that
the deterministic fallback mechanism is necessary to turn good empirical HRT
behavior into preserved HRT schedulability.

\begin{figure}[t]
    \centering
    \includegraphics[width=0.90\columnwidth]{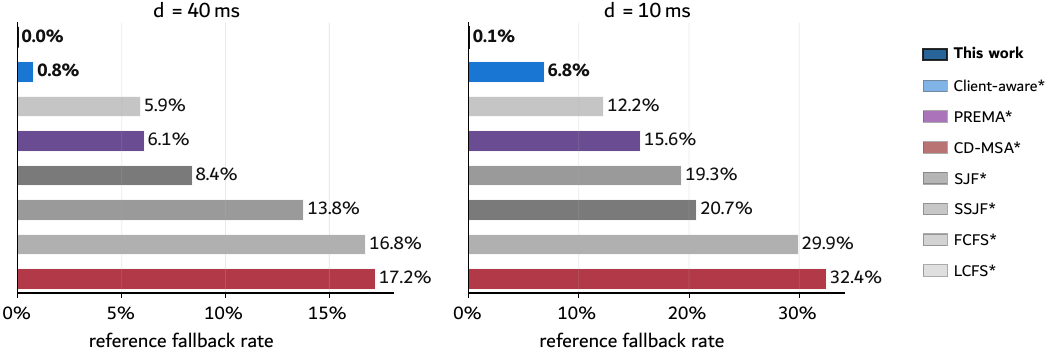}

\caption{Reference-schedule fallback rate under our HRT safety envelope.
The metric reports how often a candidate prefix is rejected and execution falls
back to the certified reference schedule \(S^{\mathrm{ref}}\). Lower is better.}
    \label{fig:reference_fallback_rate}
\end{figure}
\begin{figure}[!t]
    \centering
    \includegraphics[width=0.79\columnwidth]{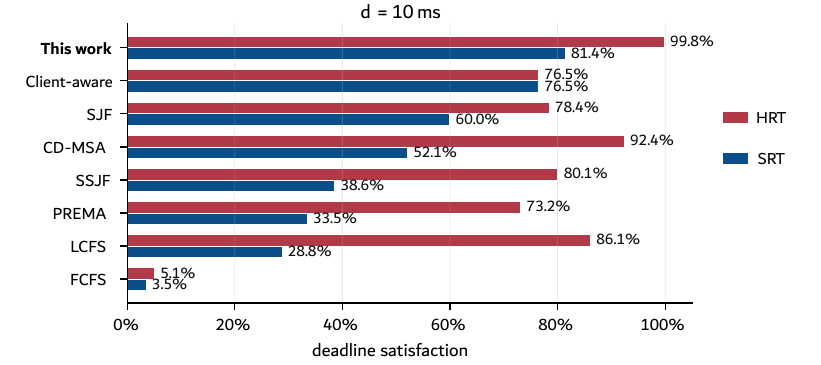}

    \caption{Ablation without the runtime HRT safety envelope at \(d=10\,\mathrm{ms}\).
    The plot reports HRT and SRT deadline satisfaction when candidate prefixes
    are executed directly, without fallback to \(S^{\mathrm{ref}}\).}
    \label{fig:no_safety_envelope_ablation}
\end{figure}

\paragraph{Scheduler overhead.}
We account for the end-to-end runtime overhead incurred at each scheduling
instant, including state encoding, policy inference, runtime certification, and
dispatching/cancellation bookkeeping. We denote this overhead by
$
\tau_{\mathrm{sch}}
=
\tau_{\mathrm{enc}}
+
\tau_{\mathrm{policy}}
+
\tau_{\mathrm{cert}}
+
\tau_{\mathrm{disp}} .
$
This term is treated as the worst-case execution time of the online scheduler
and is charged at every scheduling activation in the experimental timing
model.
The scheduler is activated every \(T_S=5\,\mathrm{ms}\). To make
\(\tau_{\mathrm{sch}}\) bounded, the number of sub-jobs visible to the online
policy is limited to a fixed maximum \(N_{\max}\): all ready HRT sub-jobs are
always included, while the SRT ready queue is truncated when necessary; more priority is given to the SRT jobs pertaining to clients with a SLO achievement rate lower than the target. This
bounds the cost of both state encoding and policy inference independently of
the total SRT backlog. For the GRU policy with hidden size 192, the measured
cost is 592 cycles per encoded sub-job on the smallest accelerator. With
\(N_{\max}=64\), this gives
$
\tau_{\mathrm{policy}} = 64 \cdot 592 = 37888 \text{ cycles}
= 37.9\,\mu\mathrm{s}
$
at 1\,GHz, i.e., \(0.76\%\) of the scheduling period. The runtime
certification step is also bounded: it scans at most all the $N$ HRT sub-jobs, therefore it has $O(N)$ complexity and is
bounded by \(N_{\max}\). Moreover, since sub-jobs are distributed across scheduling periods, $N$ is only an upper bound. Increasing the number of HRT tasks or DAG length increases $N$ linearly.

\paragraph{Scheduling period analysis.}
We evaluate the impact of $T_S$ under high pressure ($d=10$\,ms)
through two experiments: (i) deploying a policy trained with
$T_S=5$\,ms at different scheduling periods, and (ii) training
and testing a separate policy for each period.

For $T_S\in\{0.5,1,1.25,2.5,5,10\}$\,ms, the policy trained
at $5$\,ms achieves SRT hit ratios $H^S$ of
$73.73\%$, $76.84\%$, $69.45\%$, $74.96\%$, $81.40\%$, and
$71.12\%$, respectively. Although performance peaks at the
training period, the policy retains effectiveness
across the tested periods, demonstrating transfer across
scheduling frequencies.
Conversely, policies trained specifically at $0.5$, $1$,
$1.25$, $2.5$, and $10$\,ms achieve only $25.69\%$, $27.92\%$, $46.82\%$, $48.79\%$, and $62.81\%$, respectively. We attribute the poorer results of training with short periods
to the following reasons. Shorter windows contain fewer completed
sub-jobs contributing to the reward and provide less execution
time over which the policy can exploit accelerator availability.
Moreover, scheduling overhead accounts for a larger fraction
of execution time. Conversely, large periods are constrained by
timing requirements of both HRT and SRT jobs. Their relative deadlines must
accommodate waiting time, scheduling overhead, and the
required execution time. Additional timing slack also leaves
more room for online optimization.

\section{Related Work}

\begin{table}[t]

\definecolor{mygreen}{HTML}{15ba38}
\definecolor{myred}{HTML}{dd5b5b}
\definecolor{myblue}{HTML}{0298cc}

\def\sotasi{\textcolor{mygreen}{\faCheck}}
\def\sotano{\textcolor{myred}{\faClose}}
\def\sotani{\textcolor{orange}{\faMinus}}
\def\sotabo{\textcolor{myblue}{\faQuestion}}
\def\sotana{\textcolor{gray}{$\bullet$}}

\centering
\caption{Comparison of related work.}
\label{tab:sota}

\setlength{\tabcolsep}{2.5pt}
\renewcommand{\arraystretch}{1.05}

\resizebox{\columnwidth}{!}{%
\begin{tabular}{@{}lccccccccc@{}}

\toprule

\textbf{Work}
& \shortstack{\textbf{Multi-}\\\textbf{proc.}}
& \shortstack{\textbf{Unrelated}\\\textbf{machines}}
& \shortstack{\textbf{Precedence}\\\textbf{constraints}}
& \shortstack{\textbf{DNN}\\\textbf{inference}}
& \shortstack{\textbf{RL}\\\textbf{scheduler}}
& \shortstack{\textbf{Online}\\\textbf{scheduler}}
& \shortstack{\textbf{Soft}\\\textbf{real-time}}
& \shortstack{\textbf{Client}\\\textbf{aware}}
& \shortstack{\textbf{Hard}\\\textbf{real-time}}
\\

\midrule

PREMA~\cite{choi2020prema}
& \sotano & \sotana & \sotani & \sotasi & \sotano
& \sotasi & \sotasi & \sotasi & \sotano \\

MAGMA~\cite{kao2022magma}
& \sotasi & \sotasi & \sotani & \sotasi & \sotano
& \sotano & \sotano & \sotano & \sotano \\

Tang et al.~\cite{tang2011combining}
& \sotasi & \sotasi & \sotano & \sotano & \sotano
& \sotasi & \sotasi & \sotano & \sotasi \\

SSML~\cite{slackstealing}
& \sotano & \sotano & \sotano & \sotano & \sotano
& \sotasi & \sotasi & \sotano & \sotasi \\

Willemsen et al.~\cite{willemsen2025transfer}
& \sotano & \sotana & \sotano & \sotano & \sotano
& \sotasi & \sotano & \sotana & \sotasi \\

CD-MSA~\cite{wang2023cdmsa}
& \sotani & \sotano & \sotani & \sotasi & \sotano
& \sotasi & \sotasi & \sotani & \sotano \\

RELMAS~\cite{RELMAS1}
& \sotasi & \sotasi & \sotasi & \sotasi & \sotasi
& \sotasi & \sotasi & \sotano & \sotano \\

Client-aware~\cite{RELMAS2}
& \sotasi & \sotasi & \sotasi & \sotasi & \sotasi
& \sotasi & \sotasi & \sotasi & \sotano \\

\midrule

\textbf{Ours}
& \sotasi & \sotasi & \sotasi & \sotasi & \sotasi
& \sotasi & \sotasi & \sotasi & \sotasi \\

\bottomrule
\addlinespace[0.5mm]

\multicolumn{10}{r}{%
\footnotesize
\sotasi\, Yes
\quad
\sotano\, No
\quad
\sotani\, Limited
\quad
\sotana\, Not Applicable
} \\

\end{tabular}%
}

\end{table}

Table~\ref{tab:sota} compares representative related work along key dimensions. 

\paragraph{Hard-periodic/soft-aperiodic scheduling on heterogeneous resources.}
Tang et al.~\cite{tang2011combining} study the joint scheduling of hard
periodic and soft aperiodic tasks on heterogeneous resources. Their approach is divided into two phases: an offline phase, where hard periodic tasks are scheduled using either Heterogeneous EDF (H-EDF) or
Fastest First (FF); then, online, soft aperiodic tasks are scheduled using the per-resource slack windows extracted from the resulting
hard schedule, using Resource-Centric
(RC), Task-Centric (TC), or Task-Centric with Slack Defragmentation (TCSD).
TCSD can shift hard periodic instances earlier than their latest start
times to enlarge later slack windows. The original model assumes independent,
non-preemptive tasks that can execute on any resource with resource-dependent
execution times.

\paragraph{DNN inference on domain-specific accelerators.}
Extensive research has investigated scheduling strategies for DNN inference on accelerator-based platforms. Static solutions such as LayerWeaver~\cite{oh2021layerweaver} and AI-MT~\cite{baek2020multi} assign DNN layers to accelerators based on analytical models, but they cannot adapt to runtime variability. Schedulers like PREMA~\cite{choi2020prema} and MAGMA~\cite{kao2022magma} adjust accelerator allocation in response to workload fluctuations. Techniques such as Planaria~\cite{ghodrati2020planaria} and MoCA~\cite{kim2023moca} redistribute compute resources at finer granularity to improve utilization and co-location efficiency. Orion~\cite{strati2024orion} considers multiple priority queues in GPU scheduling. The considered systems optimize different system-level metrics, without concentrating on single client QoS, and without providing deterministic HRT guarantees. CD-MSA~\cite{wang2023cdmsa} targets EDF-based real-time DNN inference on multi-systolic-array accelerators, which are homogeneous accelerator platforms.
In parallel, several works have focused on enhancing the time predictability of DNNs on general-purpose processors, including efforts to improve timing isolation and scheduling analysis~\cite{casini2020timing}.

\paragraph{Slack stealing and modified look-ahead EDF.}
Slack stealing is a classical mechanism for improving the responsiveness of
soft aperiodic work in systems where hard periodic tasks have already been
guaranteed. A slack stealer exploits unused processing time from the periodic schedule to execute the aperiodic work, while ensuring that previously guaranteed
periodic deadlines are not violated~\cite{lehoczky1992optimal}. This idea is
closely related to the look-ahead EDF (laEDF) mechanisms used in real-time dynamic
voltage scaling: Pillai and Shin~\cite{pillai2001rtDVS} compute the minimum
amount of periodic work that must execute before the next deadline, deferring
the remaining work to later intervals when this is safe. SSML~\cite{slackstealing}
adapts laEDF principle to hard-periodic and soft-aperiodic
scheduling for uniprocessor systems with independent tasks.

\section{Conclusion}

This paper presented a scheduling framework for heterogeneous multi-accelerator systems executing both hard real-time and soft real-time QoS-aware DNN inference workloads. The proposed approach combines an online RL-based scheduler with a deterministic runtime certification mechanism that guarantees HRT correctness independently of the learned policy. To enable safe runtime adaptation, we introduced a schedulability-preserving prefix replacement criterion for heterogeneous precedence-constrained schedules. Experimental results showed that the framework preserves HRT schedulability while improving QoS-aware deadline satisfaction, client-level service quality, and accelerator utilization compared to state-of-the-art scheduling approaches. Future work includes extending the framework to support additional runtime uncertainties, workload prediction mechanisms, and alternative QoS optimization objectives.

\subsubsection*{Acknowledgements.}
This work was supported by the European Commission through the Chips Joint Undertaking under Grant Agreement No.~101139790 (ECS4DRES).

\bibliographystyle{splncs04}
\bibliography{bibliography}


\end{document}